\documentclass[reqno]{amsart}

\usepackage{mathtools,amsmath,amsthm,amssymb,amsfonts,enumerate,graphicx,color,pstricks}
\usepackage[T1]{fontenc} 
\numberwithin{equation}{section} 
\theoremstyle{plain}
\newtheorem{theo+}           {Theorem}      [section]
\newtheorem{prop+}  [theo+]  {Proposition}
\newtheorem{coro+}  [theo+]  {Corollary}
\newtheorem{lemm+}  [theo+]  {Lemma}
\newtheorem{defi+}  [theo+]  {Definition}
\newtheorem{conj+}  [theo+]  {Conjecture}

\theoremstyle{definition}
\newtheorem{rema+}  [theo+]  {Remark}
\newtheorem{prob+}  [theo+]  {Problem}
\newtheorem{exam+}  [theo+]  {Example}

\newenvironment{theorem}{\begin{theo+}}{\end{theo+}}
\newenvironment{proposition}{\begin{prop+}}{\end{prop+}}

\newenvironment{lemma}{\begin{lemm+}}{\end{lemm+}}

\newcommand{\om}{\omega}
\newcommand{\tha}{\theta}
\newcommand{\Om}{\Omega}
\newcommand{\ti}{\textup i}

\newcommand{\id}{\operatorname{id}}
\begin{document}

\baselineskip 18pt
\larger[2]
\title[Special polynomials related to the eight-vertex model II]
{Special polynomials related to the supersymmetric eight-vertex model. II. Schr\"odinger equation.} 
\author{Hjalmar Rosengren}
\address
{Department of Mathematical Sciences
\\ Chalmers University of Technology and University of Gothenburg\\SE-412~96 G\"oteborg, Sweden}
\email{hjalmar@chalmers.se}
\urladdr{http://www.math.chalmers.se/{\textasciitilde}hjalmar}

\thanks{Research  supported by the Swedish Science Research
Council (Vetenskapsr\aa det)}

\begin{abstract}
We show that   symmetric polynomials
previously introduced by the author satisfy a certain differential equation.
After a change of variables, it can be written as
 a  non-stationary Schr\"odinger equation with elliptic potential, which is closely related to the
 Knizhnik--Zamolodchikov--Bernard equation and to the canonical quantization of the Painlev\'e VI equation. In a subsequent paper, this will be used to construct a four-dimensional lattice of tau functions for Painlev\'e VI. 
\end{abstract}

\maketitle        


\section{Introduction}

The present work is the second part of a series, devoted to the study 
of certain  symmetric polynomials related to the eight-vertex model and other elliptic solvable lattice models of statistical mechanics.
In the first part \cite{r1a}, we introduced these polynomials and 
studied their behaviour at special parameter values corresponding to cusps of the relevant modular group $\Gamma_0(12)$. In the present work, we continue this study by proving that our polynomials solve a non-stationary Schr\"odinger equation with elliptic potential. 

To be more precise, let $m$ be a non-negative integer and $\mathbf k\in\mathbb Z^4$ be such that $|\mathbf k|+m=2n$ is even (throughout, $|\mathbf k|=\sum_j k_j$). In \cite{r1a}, we introduced a certain space $\Theta_n^{\mathbf k}$ 
of quasi-periodic meromorphic functions; see \S \ref{ps}. 
We proved that $\dim\Theta_n^{\mathbf k}=m$, and constructed 
explicit symmetric rational functions $T_n^{(\mathbf k)}$ of $m$ variables such that, up to an elementary factor and a change of variables, 
$T_n^{(\mathbf k)}$ spans the one-dimensional space  $(\Theta_n^{\mathbf k})^{\wedge m}$. Since the denominator in $T_n^{(\mathbf k)}$ is elementary, they are essentially symmetric polynomials. 

The functions $T_n^{(\mathbf k)}$ include as special cases various polynomials related (sometimes conjecturally) to elliptic lattice models of statistical mechanics, at the  parameter values $\Delta=\pm 1/2$. Indeed, they appear as the ground state eigenvalue for the $Q$-operator of the eight-vertex model \cite{bm1,bm2}, 
in expressions for 
the domain wall partition function of the eight-vertex-solid-on-solid and three-colour models 
\cite{r0,r} and in expressions for ground state eigenvectors of the  XYZ spin chain \cite{bm4,ras,zj} and related chains \cite{bh,fh,h}.

In the present paper, we  show that the elements in the space  
$(\Theta_n^{\mathbf k})^{\wedge m}$  satisfy
 a non-stationary Schr\"odinger equation
with elliptic potential, see Theorem \ref{semt}. 
When $m=1$,  this equation takes the form
\begin{equation}\label{nse}\psi_t=\frac 12\,\psi_{xx}-V\psi, \end{equation}
where  $V$ is the Darboux potential \cite{d,i,v}
$$V(x,t)=\sum_{j=0}^3\frac{k_j(k_j+1)}{2}\,\wp(x-\gamma_j|1,2\pi\ti t), $$
with $\gamma_j$  the four half-periods of the $\wp$-function.
The $m$-variable case is simply the equation for $m$ non-interacting
particles with the same potential.
The case $m=1$, $\mathbf k=(0,n,n,-1)$ corresponds to
 the non-stationary Lam\'e equation in \cite{bm1}.

The equation \eqref{nse} has appeared in the literature in several contexts.
It is the canonical quantization
of  Painlev\'e VI, and has been studied from this viewpoint by Nagoya \cite{na,na2}, Suleimanov \cite{su,su2} and Zabrodin and Zotov \cite{z}, see
\cite{no,zs} for related work. 
To explain this, recall the elliptic form of  Painlev\'e VI,
$$\frac{d^2 q}{dt^2}=\sum_{j=0}^3\nu_j\wp'(q-\gamma_j|1,2\pi\ti t). $$
It  is  equivalent to the Hamiltonian system
$$\frac{d q}{d t}=\frac{\partial H}{\partial p},\qquad \frac{d p}{d t}=-\frac{\partial H}{\partial q}, $$
where
$$H=\frac{p^2}{2}-\sum_{j=0}^3\nu_j\wp(q-\gamma_j|1,2\pi\ti t). $$
In imaginary time, 
the canonical quantization of this system is the \emph{quantum} Painlev\'e VI
equation
\begin{equation}\label{cqs}\hbar\psi_t=\frac 12\,{\psi_{xx}}-\sum_{j=0}^3\nu_j\wp(x-\gamma_j|1,2\pi\ti t)\psi,
\end{equation}
which for $\hbar=1$ reduces to  \eqref{nse}, with $\nu_j=k_j(k_j+1)/2$.

The equation \eqref{cqs} also appears in conformal field theory 
and the representation theory of affine Lie algebras. 
At least under some extra condition on the parameters, it is the one-dimensional case of the Knizhnik--Zamolodchikov--Bernard heat equation satisfied by conformal blocks of Wess--Zumino--Witten theory on a torus \cite{b,ek}. 
The general case also appears in conformal field theory \cite{fl}.
Recently, Kolb  \cite{ko} identified
the corresponding Schr\"odinger operator with the radial part of the Casimir operator for $\widehat{sl}(2)$ with respect to zonal spherical functions. Interestingly,
the condition $\hbar =1$ corresponds to central charge $c=1$,  a case  known to have close connections to Painlev\'e VI, see e.g.\ \cite{er,gil}. 
Finally, 
we mention the recent paper \cite{lt}, where a more general equation, representing interacting particles, is used to study the Inozemtsev model.

By a change of variables, the Schr\"odinger equation can be transformed
to an algebraic differential equation for the functions $T_n^{(\mathbf k)}$, see
 Theorem \ref{pdet}. Special cases of this equation have been obtained by Bazhanov and Mangazeev \cite{bm1,bm4} (without complete proof) and  Zinn-Justin \cite{zj}. 

An important application of the Schr\"odinger equation is that,
when combined with minor relations for the determinants defining $T_n^{(\mathbf k)}$, it can be used to derive bilinear relations for the polynomials.
Although many such relations exist, in the present paper we just give two examples, see Theorem \ref{rt}.

In the next paper of this series   \cite{r2}, 
 Theorem \ref{rt} will be used to identify  the case $m=0$ of
$T_n^{(\mathbf k)}$ with tau functions of Painlev\'e VI, obtained from one of Picard's solutions by acting with the full four-dimensional lattice of B\"acklund transformations. These tau functions can be obtained from
$m=1$ instances of $T_n^{(\mathbf k)}$, that is, from solutions
to \eqref{nse}, by specializing the variable to a half period. 
(More precisely, the solutions to
 \eqref{nse} that we construct
 satisfy $\psi(x)=\mathcal O((x-\gamma_j)^{k_j+1})$;
we claim that there  is a natural rescaling of these solutions so that their
 leading behaviour at the points $x=\gamma_j$ is given  by Painlev\'e  tau functions.) A similar observation was made in \cite{na2} for another class of solutions. Presumably, this phenomenon is linked to the close
relation between \eqref{cqs} with $\hbar=1$ and the Lax representation of Painlev\'e VI described in \cite{su,z}. 

{\bf Acknowledgements:} I would like to thank Vladimir Bazhanov, Stefan Kolb, Vladimir Mangazeev, Dmitrii Novikov, Bulat Suleimanov and Paul Zinn-Justin for interesting discussions and correspondence.

\section{Preliminaries}\label{ps}

We recall some relevant facts from \cite{r1a}.
For $\tau$ fixed in the upper half-plane, we write
 $p=e^{\pi\ti\tau}$. 
 We will also write
 $\omega=e^{2\pi\ti/3}$.
 We will use the notation
\begin{align*}(x;p)_\infty&=\prod_{j=0}^\infty(1-xp^j),\\
\theta(x;p)&=(x;p)_\infty(p/x;p)_\infty.\end{align*}
Repeated variables are used as a short-hand for products; for instance,
$$\theta(a,\pm b;p)=\theta(a;p)\theta(b;p)\theta(-b;p). $$

The function
 \begin{equation}\label{pse}\psi(z)=\psi(z,\tau)=p^{\frac 1{12}}(p^2;p^2)_\infty e^{-\pi \ti z}\theta(e^{2\pi \ti z},\pm pe^{2\pi \ti z};p^2)\end{equation}
satisfies
\begin{equation}\label{pqp}\psi(z+1)=\psi(-z)=-\psi(z),\qquad  \psi(z+\tau)=e^{-3\pi \ti(\tau+2z)}\psi(z),\end{equation}
 the heat equation
\begin{equation}\label{ops}12\pi \ti\frac{\partial\psi}{\partial\tau}=\frac{\partial^2\psi}{\partial z^2}\end{equation}
and 
 \begin{equation}\label{pss}\psi(z)=\psi\left(z+\frac 13\right)+\psi\left(z-\frac 13\right). \end{equation}

We will write
\begin{align*}x(z)&=x(z,\tau)=\frac{\theta(- p\omega;p^2)^2\theta(\omega e^{\pm 2\pi \ti z};p^2)}{\theta(-\omega;p^2)^2\theta( p\omega e^{\pm 2\pi \ti z};p^2)}, \\
\zeta&=\zeta(\tau)=\frac{\omega^2\theta(-1,- p\omega;p^2)}{\theta(- 
p,-\omega;p^2)}. \end{align*}
The function $x$ generates the field of even elliptic functions with periods $(1,\tau)$. Moreover, $\tau\mapsto\zeta(2\tau)$ generates the field of modular functions for the group $\Gamma_0(12)$.

We will need the identity
\begin{equation}\label{xd} x(z)- x(w)=-\frac{\om\tha( p, p\om;p^2)\tha(- p\om
;p^2)^2}{e^{2\pi \ti z}\tha(-\om;p^2)^2}\frac{\tha(e^{2\pi \ti(z\pm w)};p^2)}{\tha( p\om e^{\pm 2\pi \ti z}, p\om e^{\pm 2\pi \ti w};p^2)}\end{equation}
and its limit case
\begin{equation}\label{xp} x'(z)=\frac{2\pi \ti\om(p^2;p^2)_\infty^2\tha( p, p\om;p^2)\tha(- p\om
;p^2)^2}{\tha(-\om;p^2)^2}\frac{e^{-2\pi \ti z}\tha(e^{4\pi \ti z};p^2)}{\tha( p\om e^{\pm 2\pi \ti z};p^2)^2}.\end{equation}
Denoting the half-periods in $\mathbb Z+\tau\mathbb Z$ by
$$\gamma_0=0,\qquad \gamma_1=\frac\tau 2,\qquad \gamma_2=\frac\tau2+\frac 12,\qquad \gamma_3=\frac 12,$$ 
the values $\xi_j=x(\gamma_j)$ and $\eta_j=x(\gamma_j+1/3)$ are given by
\begin{subequations}\label{xv}
\begin{align}
\xi_0&=2\zeta+1,& \xi_1&=\frac{\zeta}{\zeta+2},
& \xi_2&=\frac{\zeta(2\zeta+1)}{\zeta+2},
& \xi_3&=1,\\
\eta_0&=0, & \eta_1&=\infty, &\eta_2&=\frac{2\zeta+1}{\zeta+2},
& \eta_3&=\zeta.
\end{align}
\end{subequations}
Moreover \cite[Lemma 9.1]{r},
\begin{subequations}\label{zp}
 \begin{align} 
\zeta+1&=-\frac{\tha(p,-p\om;p^2)}{\tha(-p,p\om;p^2)}, \\
 \zeta-1&=\frac{\tha(p,p\om;p^2)\tha(\om;p^2)^2}{\tha(-p,-p\om;p^2)\tha(-\om;p^2)^2},\\
\zeta+2&=p\frac{\tha(-1,-\om;p^2)\tha(\om;p^2)^2}{\tha(-p,-p\om ;p^2)
 \tha(p\om;p^2)^2}, \\
 2\zeta+1&=\frac{\tha(-p\om,\om;p^2)^2}{\tha(-\om,p\om;p^2
 )^2}. 
 \end{align}
\end{subequations}


For $n\in\mathbb Z$ and $\mathbf k=(k_0,k_1,k_2,k_3)\in\mathbb Z^4$,
such that $2n\geq|\mathbf k|=\sum_jk_j$, we  define 
 a function space  $\Theta_{n}^{\mathbf k}$ as follows \cite[Lemma 2.3]{r1a}.
The elements in the space are  meromorphic functions, which are analytic
outside the lattice $\frac 16\mathbb Z+ \frac\tau 2\mathbb Z$, satisfy
\begin{subequations}\label{vde}
\begin{equation}\label{fqp}f(z+1)=f(z),\qquad  f(z+\tau)=e^{-6\pi \ti n(\tau+2z)}f(z),\qquad f(-z)=-f(z),\end{equation}
\begin{equation}\label{fe}f(z)+f\left(z+\frac 13\right)+f\left(z-\frac 13\right)=0 \end{equation}
\end{subequations}
and, for $j=0,1,2,3$, 
\begin{equation}\label{oca}\lim_{z\rightarrow \gamma_j}(z-\gamma_j)^{1-2k_j}f(z)=\lim_{z\rightarrow\gamma_j}(z-\gamma_j)^2\left(f\left(z+\frac 13\right)-f\left(z-\frac 13\right)\right)=0. \end{equation}


 Writing $m=2n-|\mathbf k|$, we have proved that $\dim\Theta_n^{\mathbf k}=m$ \cite[Thm.\ 2.4]{r1a}. Moreover, realizing the maximal exterior power
$(\Theta_n^{\mathbf k})^{\wedge m}$  as a space of anti-symmetric functions in $z_1,\dots,z_m$, it  is spanned by 
\begin{equation}\label{ukl}\prod_{j=1}^{m}\left(M_n(z_j){\prod_{l=0}^3(x_j-\xi_l)^{k_l}}\right) \Delta(x_1,\dots,x_{m})T_{n}^{(\mathbf k)}(x_1,\dots,x_{m}), \end{equation}
where 
\begin{equation}\label{m}M_n(z)=e^{-2\pi\ti z}\tha(e^{4\pi\ti z};p^2)\tha(\om pe^{\pm 2\pi\ti z};p^2)^{3n-2},\end{equation}
$x_j=x(z_j)$, $\Delta(\mathbf x)=\prod_{i<j}(x_j-x_i)$ and $T_{n}^{(\mathbf k)}$ is a certain symmetric rational function, depending also rationally on the parameter $\zeta$.

To describe  the construction of  $T_{n}^{(\mathbf k)}$, 
we start with the case
\begin{multline}\label{gdf} T_n^{(0,0,0,0)}(x_1,\dots,x_{2n})\\
=\frac{\prod_{i,j=1}^nG(x_i,x_{n+j})}{\Delta(x_1,\dots,x_n)\Delta(x_{n+1},\dots,x_{2n})}\,\det_{1\leq i,j\leq n}\left(\frac{1}{G(x_i,x_{n+j})}\right),
\end{multline}
where
$$G(x,y)=(\zeta+2)xy(x+y)-\zeta(x^2+y^2)-2(\zeta^2+3\zeta+1)xy+\zeta(2\zeta+1)(x+y).$$
Then,  $T_n^{(0,0,0,0)}$ is a symmetric polynomial in all its variables, depending also as a polynomial on the parameter $\zeta$.
Up to a change of variables, $T_n^{(0,0,0,0)}$ coincides with the polynomial $H_{2n}$ of \cite{zj}. 
For $\mathbf k\in\mathbb Z_{\geq 0}^4$, $T_n^{(\mathbf k)}$ is  obtained from
$T_n^{(0,0,0,0)}$ by specializing $k_j$ of the variables to $\xi_j$, for $0\leq j\leq 3$. If  $k_j<0$ for some $j$, the definition is more complicated.

To explain the general definition of  $T_n^{(\mathbf k)}$, let
$$(\sigma f)(z)=f(z+1/3)-f(z-1/3)$$
and 
\begin{align}\label{a}a(x)&=
\big(x-(2\zeta+1) \big)(x-1)\big(({\zeta+2})x-\zeta\big)\big(({\zeta+2})x-{(2\zeta+1)\zeta}\big),\\
\nonumber &=(\zeta+2)^2\prod_{l=0}^3(x-\xi_l).
\end{align}
Then, $a(x(z))$ has a meromorphic square root that we will denote $\sqrt{a}$. 
(In  \cite{r1a}, $\sigma$ and $\sqrt{a}$ are denoted $\sqrt 3\sigma/i$ and $\ti\phi/\sqrt 3$.) 
We can then write
\begin{equation}\label{us}\sigma\Big|_{\Theta_n^{(0,0,0,0)}}=\frac{M_n}{\sqrt a}\,\hat \sigma_n \frac 1 {M_n},\end{equation}
where $\hat\sigma_n$ is an operator acting between appropriate spaces of polynomials in $x=x(z)$. It is determined by
\begin{multline}\label{hs}\hat\sigma_n\left((x-a)\prod_{j=1}^{n-1}(x-b_j)G(x,b_j)\right)=\prod_{j=1}^{n-1}(x-b_j)G(x,b_j)\\
\times\Big(x(x-2\zeta-1)\big((\zeta+2)x-3\zeta\big)-a\big((\zeta+2)x-\zeta\big)(2\zeta+1-3x)\Big),
\end{multline}
where $a$ and $b_j$ are arbitrary.

We may now define, for $0\leq k\leq 2n$,
\begin{multline}\label{ttv}T(x_1,\dots,x_k;x_{k+1},\dots,x_{2n})\\
=\frac{(\id^{\otimes k}\otimes\,\hat\sigma_n^{\otimes (2n-k)})\Delta(x_1,\dots,x_{2n})T_n^{(0,0,0,0)}(x_1,\dots,x_{2n})}{\Delta(x_1,\dots,x_k)\Delta(x_{k+1},\dots,x_{2n})}.
 \end{multline}
The identity \eqref{hs} can be applied termwise to \eqref{gdf} to give an explicit formula for \eqref{ttv} as a block determinant, see \cite[Eq.\ (2.39)]{r1a}.

The general definition of $T_n^{(\mathbf k)}$ can now be stated as
$$T_n^{(\mathbf k)}(x_1,\dots,x_{m})=\frac{(-1)^{\binom{|\mathbf k^-|}2}\,T(x_1,\dots,x_m,\boldsymbol\xi^{{\mathbf k}^+};\boldsymbol\xi^{{\mathbf k}^-})}{2^{|\mathbf k^-|}\prod_{i,j=0}^3G(\xi_i,\xi_j)^{k_i^-k_j^+}\prod_{j=1}^m\prod_{i=0}^3G(x_j,\xi_i)^{k_i^-}},$$
where $k_j^\pm=\max(\pm k_j,0)$ and 
$$\boldsymbol\xi^{{\mathbf k}}=(\underbrace{\xi_0,\dots,\xi_0}_{k_0},\dots,\underbrace{\xi_3,\dots,\xi_3}_{k_3}).$$
 Then,
\begin{equation}\label{ten}T_n^{(\mathbf k+\mathbf l)}(x_1,\dots,x_m)=T_n^{(\mathbf k)}(x_1,\dots,x_m,\boldsymbol\xi^{\mathbf l}),\qquad \mathbf l\in\mathbb Z_{\geq 0}^4.
\end{equation}

In \S \ref{cts}, we will also need  dual functions defined by
\begin{equation}\label{unk}U_n^{(\mathbf k)}(x_1,\dots,x_{m})=\frac{(-1)^{\binom{|\mathbf k^+|}2}\,T(\boldsymbol\xi^{{\mathbf k}^+};x_1,\dots,x_m,\boldsymbol\xi^{{\mathbf k}^-})}{2^{|\mathbf k^+|}\prod_{i,j=0}^3G(\xi_i,\xi_j)^{k_i^-k_j^+}\prod_{j=1}^m\prod_{i=0}^3G(x_j,\xi_i)^{k_i^+}}.\end{equation}
They satisfy 
\begin{equation}\label{uen}U_{n-|\mathbf l|}^{(\mathbf k-\mathbf l)}(x_1,\dots,x_m)=U_n^{(\mathbf k)}(x_1,\dots,x_m,\boldsymbol\xi^{\mathbf l}),\qquad \mathbf l\in\mathbb Z_{\geq 0}^4,
\end{equation}
which can be proved similarly as \eqref{ten}.
It follows easily from \eqref{ttv} that
\begin{multline}\label{stu}\hat\sigma_{n+|\mathbf k^-|}^{\otimes m}\prod_{j=1}^m\prod_{i=0}^3(x_j-\xi_i)^{k_i^+}G(x_j,\xi_i)^{k_i^-}
(\Delta T_n^{(\mathbf k)})(x_1,\dots,x_m)\\
=(-1)^{\binom{|\mathbf k|}2}2^{|\mathbf k|} 
\prod_{j=1}^m\prod_{i=0}^3(x_j-\xi_i)^{k_i^-}G(x_j,\xi_i)^{k_i^+}
(\Delta U_n^{(\mathbf k)})(x_1,\dots,x_m).
\end{multline}
By \cite[Prop.\ 2.20]{r1a}, we have up to an explicit 
factor independent of the variables $x_j$,
\begin{align*}U_n^{(k_0,k_1,k_2,k_3)}(x_1,\dots,x_m)&\sim\prod_{j=1}^m
x_j\left(x_j-\frac{2\zeta+1}{\zeta+2}\right)(x_j-\zeta)
\\
&\quad\times T_{m-2-n}^{(-k_0-1,-k_1-1,-k_2-1,-k_3-1)}(x_1,\dots,x_m); \end{align*}
however, we will not need this fact.

\section{Schr\"odinger  equation}

\subsection{Schr\"odinger equation with elliptic potential}

In this Section, we show that the elements in the one-dimensional space
 $(\Theta_n^{(\mathbf k)})^{\wedge m}$ 
 satisfy a Schr\"odinger equation with elliptic potential. We let $\wp=\wp(z|1/3,\tau)$ denote 
 Weierstrass's 
$\wp$-function as defined in \cite{ww}. It is an even elliptic function with periods $1/3$ and $\tau$, with no singularities except double poles at the lattice points, such that
\begin{equation}\label{wpr}\lim_{z\rightarrow 0} z^2\wp(z)=1. \end{equation}
These properties determine  $\wp$ uniquely up to an additive constant, whose value is irrelevant for our purposes.




\begin{theorem}\label{semt}
Let $\Psi(z_1,\dots,z_m,\tau)$ be a meromorphic function, which for fixed $\tau$ belongs to $(\Theta_n^{\mathbf k})^{\wedge m}$, and let
 \begin{multline*}\Phi=\prod_{j=1}^m\Big(\left
(e^{-3\pi \ti z_j}\tha(e^{6\pi\ti z_j};p^6)\right)^{k_0}
\tha(p^3e^{6\pi\ti z_j};p^6)^{k_1}\\
\times\tha(-p^3e^{6\pi\ti z_j};p^6)^{k_2}
\left(e^{-3\pi \ti z_j}\tha(-e^{6\pi\ti z_j};p^6)\right)^{k_3}\Big)
.
\end{multline*}
 Then, $\Phi^{-1}\Psi$ satisfies the Schr\"odinger equation
\begin{equation}\label{sem}\mathcal H\Phi^{-1}\Psi=C\Phi^{-1}\Psi, \end{equation}
where
\begin{equation}\label{om}\mathcal H= -12\pi \ti m \frac{\partial}{\partial\tau}+\sum_{j=1}^m \left(\frac{\partial^2}{\partial z_j^2}-V(z_j)\right),\end{equation}
 $C$ is  independent of the variables $z_j$ and 
$$V(z)=\sum_{j=0}^3k_j(k_j+1)\wp(z-\gamma_j). $$
\end{theorem}

Note that  $\Psi$ is only determined up to a factor depending on $\tau$; the
factor $C$  depends on this choice of normalization.
If we choose  $C=0$ and use that 
$\wp(z|1/3,\tau)=9\wp(3z|1,3\tau)$, we find that 
the case $m=1$ of \eqref{sem} reduces to \eqref{nse}, with
 $z=x/3$, $\tau=2\pi\ti t/3$.

For the proof of Theorem \ref{semt}, we first state
 the following elementary consequence of the chain rule.

\begin{lemma}\label{htl}
If $f(z,\tau)$ is a meromorphic function in two variables satisfying
$$f(z+\tau,\tau)=\varepsilon e^{-\lambda(\tau+2z)}f(z,\tau), $$
where $\varepsilon$ and $\lambda$ are arbitrary constants, 
then the same identity holds with $f$ replaced by
$$\frac{\partial^2 f}{\partial z^2}-4\lambda\frac{\partial f}{\partial \tau}.$$
\end{lemma}

Let us  express the potential in terms of the function
\begin{equation}\label{phi}\phi(z)=\frac \ti{3\pi}\frac{(-p^6;p^6)_\infty^2}{(p^6;p^6)_\infty^2}\frac{\tha(e^{6\pi\ti z};p^6)}{\tha(-e^{6\pi\ti z};p^6)}. 
\end{equation}
It is easy to see that \cite[\S 20.53, Example 1]{ww}
$$\wp(z)-\wp(1/6)=\frac 1{\phi(z)^2}. $$
Thus, up to a change of the constant $C$, 
we may as well prove that \eqref{sem} holds with the modified potential
\begin{equation}\label{pot}V(z)=\sum_{j=0}^3\frac{k_j(k_j+1)}{\phi(z-\gamma_j)^2}. \end{equation}
Note that \eqref{wpr}
translates to
\begin{equation}\label{phd}\phi'(0)^2=1 \end{equation}
(indeed, one may check directly from \eqref{phi} that $\phi'(0)=1$).

Since, by \cite[Thm.\ 2.4]{r1a}, $\dim(\Theta_n^{\mathbf k})^{\wedge m}=1$, it is enough to show that 
$\Xi=\Phi\mathcal H\Phi^{-1}\Psi\in(\Theta_n^{\mathbf k})^{\wedge m}$.
As a function of each $z_j$, $\Phi$ satisfies
\begin{align}\label{php}\Phi(z+1/3)&=(-1)^{k_0+k_3}\Phi(z),\\
\nonumber\Phi(-z)&=(-1)^{k_0}\Phi(z),\\
\nonumber\Phi(z+\tau)&=(-1)^{k_0+k_1}e^{-3\pi\ti(k_0+k_1+k_2+k_3)(\tau+2z)}\Phi(z).
\end{align}
Since  $\Psi$ satisfies  \eqref{vde},
it follows that
\begin{gather*}
\begin{split}
(\Phi^{-1}\Psi)(z+1)&=(-1)^{k_0+k_3}(\Phi^{-1}\Psi)(z),\\
(\Phi^{-1}\Psi)(-z)&=(-1)^{k_0+1}(\Phi^{-1}\Psi)(z),\\
(\Phi^{-1}\Psi)(z+\tau)&=(-1)^{k_0+k_1}e^{-3\pi \ti m(\tau+2z)}(\Phi^{-1}\Psi)(z),
\end{split}\\
\nonumber(\Phi^{-1}\Psi)(z)+(-1)^{k_0+k_3}\left((\Phi^{-1}\Psi)(z+1/3)+(\Phi^{-1}\Psi)(z-1/3)\right)=0.
\end{gather*}
We must show that these  relations are  preserved by $\mathcal H$.
Since $V$ is an even elliptic function with periods $1/3$ and $\tau$, 
this is clear except for the third relation, which is covered by Lemma \ref{htl}.
 Thus, $\Xi$ satisfies  \eqref{vde}
as a function of each $z_j$. It is also obviously antisymmetric. 

 It remains to show that
\begin{equation}\label{oxa}\lim_{z_1\rightarrow \gamma_j}\phi(z_1-\gamma_j)^{1-2k_j}\Xi(z_1)=0, \end{equation}
\begin{equation}\label{oxb}\lim_{z_1\rightarrow \gamma_j}\phi(z_1-\gamma_j)^{2}\left(\Xi\left(z_1+\frac 13\right)-\Xi\left(z_1-\frac13\right)\right)=0. \end{equation}

To prove \eqref{oxa}, we  write
\begin{equation}\label{ppf}\Phi^{-1}\Psi
=\phi(z_1-\gamma_j)^{k_j+1}\cdot\frac{\Phi^{-1}\Psi}{\phi(z_1-\gamma_j)^{k_j+1}},
\end{equation}
where, by \eqref{oca} for $f=\Psi(z_1)$, the second factor is regular at $z_1=\gamma_j$.
Acting with $\mathcal H$, only the term $\partial^2/\partial z_1^2-V(z_1)$  contributes to \eqref{oxa}.
Moreover,  both derivatives must hit the first factor  in \eqref{ppf}, which is then reduced to 
$$k_j(k_j+1)\phi'(z_1-\gamma_j)^2\phi(z_1-\gamma_j)^{k_j-1}.$$
 Thus,  \eqref{oxa} follows from 
\begin{equation}\label{phdd}\lim_{z\rightarrow \gamma_j}\left(k_j(k_j+1)\phi'(z-\gamma_j)^2-V(z)\phi(z-\gamma_j)^2\right)=0,
 \end{equation}
which is true in view of \eqref{phd}.

The proof of \eqref{oxb} is similar. We start from the factorization
$$\frac{\Psi(z_1+1/3)-\Psi(z_1-1/3)}{\Phi(z)}
=\phi(z_1-\gamma_j)^{-k_j}\cdot \frac{\Psi(z_1+1/3)-\Psi(z_1-1/3)}{\Phi(z)\phi(z_1-\gamma_j)^{-k_j}}.$$
In view of \eqref{php} and the fact that $V$ is $1/3$-periodic, the operator
$\Phi\mathcal H\Phi^{-1}$ commutes with translations by $1/3$.
Using this fact,  \eqref{oxb} can be reduced to \eqref{phdd}.

\subsection{Uniformized Schr\"odinger equation}\label{pds}

The following result is a uniformized version of
Theorem \ref{semt}.
The special case $m=1$, $\mathbf k=(0,n,n,-1)$,  is equivalent to  \cite[Eq.\ (27)]{bm1} (given there without a complete proof, since 
it was not known at the time that $\dim\Theta_n^{(0,n,n,-1)}=1$).
The case $m=1$, $\mathbf k=(n,n,0,-1)$ was conjectured in \cite{bm4};
it is in fact equivalent to the case $\mathbf k=(0,n,n,-1)$
by the symmetries \cite[Cor.\ 2.19]{r1a}. 
Moreover, the case $m=2n$,  $\mathbf k=(0,0,0,0)$ is equivalent to \cite[Eq.\ (50)]{zj}.

\begin{theorem}\label{pdet}
The function $T_n^{(\mathbf k)}$ satisfies the differential equation
\begin{multline*}\left(\sum_{j=1}^m \left(a(x_j)\frac{\partial^2}{\partial x_j^2}+b(x_j)\frac{\partial}{\partial x_j}+c(x_j)\right)+m\, d\frac{\partial}{\partial\zeta}\right)\\
 f\prod_{j=1}^mF(x_j)\Delta(x_1,\dots,x_m) T_n^{(\mathbf k)}(x_1,\dots,x_m)=0,\end{multline*}
where 
$$F(x)=\prod_{j=0}^3\left(\frac{x-\xi_j}{G(x,\xi_j)}\right)^{\frac{k_j}2}, $$
\begin{align*} f&=\frac{(\zeta+1)^{\frac 14k_2(k_2+2)+\frac 14k_3(k_3+2)-(k_0+k_1)(k_2+k_3-1)+\frac 12k_2k_3}}{\zeta^{\frac 14k_1(k_1-3)+\frac 14k_2(k_2-3)+(k_0+k_3)(k_1+k_2-\frac 14)-\frac 12 k_0k_3}}\\
&\quad\times\frac {(\zeta+2)^{-\frac 34k_0(k_0+1)-\frac 14k_1(2k_1+5)-\frac 14k_2(2k_2+5)-\frac 34k_3(k_3+1)+\frac 12(k_0+k_3)(k_1+k_2)+\frac 12k_1k_2}}{(\zeta-1)^{\frac 14(k_2+k_3)(k_2+k_3-2)}(2\zeta+1)^{\frac 14(k_0+k_2)(k_0+k_2-2)} },
\end{align*}
$a(x)$ is given by \eqref{a},
$b(x)$ is a polynomial in $(x,\zeta)$ of bidegree $(3,3)$, which 
we give in terms of  the partial fraction decomposition
 \begin{align}\notag
 \frac{b(x)}{a(x)} &=
\frac{3(\zeta+1)+m(\zeta-1)(\zeta+2)}{2(\zeta+1)(x-2\zeta-1)}
+\frac{(\zeta+2)\big(3\zeta(\zeta+1)-m(2\zeta+1)(\zeta-1)\big)}{2\zeta(\zeta+1)\big((\zeta+2)x-\zeta\big)}
\\
 \label{b}& \quad+\frac{(\zeta+2)\big(3\zeta-m(\zeta^2+4\zeta+1)\big)}{2\zeta\big((\zeta+2)x-\zeta(2\zeta+1)\big)}+\frac{3}{2(x-1)}, \end{align}
$c(x)=c_0(x)+W(x)$, with
\allowdisplaybreaks{
\begin{align}
\notag c_0(x)&=\frac {3(m-2)(3m-4)}4(\zeta+2)^2x^2\\
\notag&\quad-\frac {3m-4}2(\zeta+2)\big(2(2m-3)(\zeta^2+1)+(7m-12)\zeta\big)x\\
\notag&\quad-\frac {2(2m^2-5)}3\,\zeta^4-\frac {7m^2+66m-112}6\,\zeta^3+\frac {7(m-2)(7m-8)}4\,\zeta^2\\
\label{cn}&\quad+\frac{(5m-8)(19m-14)}6\,\zeta+\frac {11m^2-24m+10}3,\\
\notag W(x)&=-k_0(k_0+1)(2\zeta+1)^3
\frac{(x-1)(x-\zeta)^2}{x^2(x-(2\zeta+1))}\\
\notag&\quad-k_1(k_1+1)\frac{\big((\zeta+2)x-{\zeta(2\zeta+1)}\big)\big((\zeta+2)x-({2\zeta+1})\big)^2}{(\zeta+2)x-{\zeta}}\\
\notag&\quad+k_2(k_2+1)\frac{(\zeta+1)(\zeta-1)^3(2\zeta+1)^3\big((\zeta+2)x-{\zeta}\big)}{\big((\zeta+2)x-{\zeta(2\zeta+1)}\big)\big((\zeta+2)x-({2\zeta+1})\big)^2}\\
\label{wp}&\quad+k_3(k_3+1)(\zeta+1)(\zeta-1)^3\frac{x^2(x-(2\zeta+1))}{(x-1)(x-\zeta)^2} \end{align}
}
and
$$ d=2\zeta(\zeta-1)(\zeta+1)(\zeta+2)(2\zeta+1).$$
\end{theorem}

The factors $f$ and $F(x_j)$ have been introduced in order to simplify the 
expressions for the coefficients $b$ and $c$. 
For later use, we note that 
commuting them
across the differential operator leads to 
\begin{equation}\label{usf}\left(\sum_{j=1}^m \left(a(x_j)\frac{\partial^2}{\partial x_j^2}+b_F(x_j)\frac{\partial}{\partial x_j}+c_F(x_j)\right)+m\, d\frac{\partial}{\partial\zeta}+e\right)
\Delta T_n^{(\mathbf k)}=0,\end{equation}
where
\begin{align}\notag b_F&=2a\frac{\partial F/\partial x}{F}+b, \\
\notag c_F&=a\frac{\partial^2 F/\partial x^2}{F}+b\frac{\partial F/\partial x}{F}+c
+md\frac{\partial F/\partial \zeta}{F},\\
\label{ee}e&=md\frac{f'}{f}.\end{align}

To prove  Theorem \ref{pdet}, we first note that, up to a  factor independent of the variables $z_j$,  
 \eqref{ukl} is proportional to
\begin{equation}\label{eps}\Phi\prod_{j=1}^m\psi(z_j)^mE(x_j)F(x_j)\,\Delta(x_1,\dots,x_m)T_n^{(\mathbf k)}(x_1,\dots,x_m),\end{equation}
where
$$E(x)
=(x-\xi_0)^{\frac{1-m}2}\left(x-\xi_1\right)^{\frac {1-m}2}\left(x-\xi_2\right)^{\frac{1-m}2}(x-\xi_3)^{\frac 12}. $$
This can be seen either from the fact that the quotient of \eqref{ukl} and \eqref{eps} 
is periodic without zeroes or poles, or  using \eqref{xd}.
We  choose the function $\Psi$ in Theorem \ref{semt} as \eqref{eps}.

We first study the action of the Schr\"odinger operator $\mathcal H$ on  $\psi(z_j)^m$. We have
\begin{multline*}\prod_{j=1}^m\psi(z_j)^{-m}\mathcal H\prod_{j=1}^m\psi(z_j)^m=\sum_{j=1}^m \left(m\frac{\psi''(z_j)}{\psi(z_j)}+m(m-1)\frac{\psi'(z_j)^2}{\psi(z_j)^2}-12\pi \ti m^2\frac{\dot\psi(z_j)}{\psi(z_j)}\right)\\
=\sum_{j=1}^mm(m-1)\left(\frac{\psi'(z_j)^2}{\psi(z_j)^2}-\frac{\psi''(z_j)}{\psi(z_j)}\right)=-\sum_{j=1}^mm(m-1)(\log\psi(z_j))'', \end{multline*}
where  $\psi'=\partial \psi/\partial z$, $\dot\psi=\partial\psi/\partial\tau$ and we used \eqref{ops} in the second step.
It follows that, for any function $X=X(z_1,\dots,z_m,\tau)$,
\begin{multline*}\prod_{j=1}^m\psi(z_j)^{-m}\mathcal H\prod_{j=1}^m\psi(z_j)^mX
=
-12\pi \ti m \dot X\\
+\sum_{j=1}^m\left(\frac{\partial^2 X}{\partial z_j^2}
+2m(\log\psi(z_j))'\frac{\partial X}{\partial z_j}-\big(m(m-1)(\log\psi(z_j))''+V(z_j)\big)X\right).
 \end{multline*}
Thus, if  $X$ can be expressed in terms of the variables $x_j=x(z_j,\tau)$ and $\zeta=\zeta(\tau)$, 
\begin{multline}\label{pu}\prod_{j=1}^m\psi(z_j)^{-m}\mathcal H\prod_{j=1}^m\psi(z_j)^mX\\
=
-12\pi \ti m \dot\zeta\frac{\partial X}{\partial\zeta}
+\sum_{j=1}^m\bigg((x_j')^2\frac{\partial^2 X}{\partial x_j^2}
+\Big(x_j''+2m(\log\psi(z_j))'x_j'-12\pi \ti m\dot x_j\Big)\frac{\partial X}{\partial x_j}\\
-\Big(m(m-1)(\log\psi(z_j))''+V(z_j)\Big)X\bigg).
 \end{multline}

We must express the coefficients in \eqref{pu} in terms of the variables $x_j$ and $\zeta$. We formulate the relevant elliptic function identities as a series of lemmas.
It is convenient to introduce the parameter
$$\chi=\chi(\tau)=4\pi^2p(p^2;p^2)_\infty^4
\theta(-1;p^2)\theta(-\om;p^2)^3. $$

\begin{lemma}\label{ddl}
We have
\begin{equation}\label{xds} (x')^2=-\frac{\chi}{2\zeta(\zeta+1)(\zeta+2)}\,a,\end{equation}
\begin{equation}\label{xsd} x''=-\frac{\chi}{4\zeta(\zeta+1)(\zeta+2)}\frac{\partial a}{\partial x},\end{equation}
where $a$ is as in \eqref{a}.
\end{lemma}

\begin{proof}
Using \eqref{xd} and \eqref{xp}, we can write
\begin{equation*}\begin{split} x'(z)^2&= \frac{4\pi^2 \om^2(p^2;p^2)_\infty^4\tha( p, p\om;p^2)^2\tha(- p\om
;p^2)^4}{\tha(-\om;p^2)^4}\\
&\quad\times\frac{\theta(e^{\pm 2\pi \ti z},-e^{\pm 2\pi \ti z},pe^{\pm 2\pi \ti z},-pe^{\pm 2\pi \ti z};p^2)}{\tha( p\om e^{\pm 2\pi \ti z};p^2)^4}\\
&=-\frac{4\pi^2p^2\om^2(p^2;p^2)_\infty^4\tha(-\om;p^2)^6\tha(\om;p^2)^2}{\tha(p;p^2)^2\tha(-p\om;p^2)^2(\zeta+2)^2}\,a( x).
\end{split}\end{equation*}
By \eqref{zp}, this can be written in the form  \eqref{xds}, which is then differentiated to yield \eqref{xsd}.
\end{proof}

It follows from  Lemma \ref{ddl} that
\begin{equation}\label{sd}\frac{\partial ^2}{\partial z^2}=-\frac{\chi}{4\zeta(\zeta+1)(\zeta+2)}\left(2a\frac{\partial ^2}{\partial x^2}+\frac{\partial a}{\partial x}\frac{\partial}{\partial x}\right). \end{equation}

\begin{lemma}
The function $\psi$ satisfies
\begin{align}
\label{ld}\frac{\psi'(1/3)}{\psi(1/3)}&=-2\pi \ti\frac{(p^2;p^2)_\infty^2\tha(-\om;p^2)}{\tha(-1,\om;p^2)}, \\
\label{ld2}\frac{\psi'(1/3+\tau/2)}{\psi(1/3+\tau/2)}+3\pi \ti&=-2\pi \ti\om^2\frac{(p^2;p^2)_\infty^2\tha(-p\om;p^2)}{\tha(-p,\om;p^2)}. \end{align}
\end{lemma}

\begin{proof}
Since  $\psi$ is odd, differentiating \eqref{pss}
gives $\psi'(1/3)=\psi'(0)/2$. On the other hand, by
\eqref{pse} we can write
$$\frac{\psi(z)}{1-e^{4\pi \ti z}}=\frac{p^{\frac 1{12}}e^{-\pi \ti z}(p^2,p^2e^{\pm 4\pi \ti z};p^2)_\infty}{\tha(-e^{2\pi \ti z};p^2)_\infty}, $$
which for $z\rightarrow 0$ reduces to
\begin{equation}\label{pdn}\psi'(0)=-4\pi \ti \frac{p^{\frac 1{12}}(p^2;p^2)_\infty^3}{\tha(-1;p^2)}. \end{equation}
After simplification, this yields \eqref{ld}.

Similarly, it follows from \eqref{pqp} and \eqref{pss} that
$$\psi(z)=\psi\left(z+\frac 13\right)-e^{3\pi \ti(\tau-2z)}\psi\left(-z+\tau+\frac 13\right). $$
Differentiating this identity and letting $z=\tau/2$ gives
\begin{equation}\label{psd}\frac{\psi'(1/3+\tau/2)}{\psi(1/3+\tau/2)}+3\pi \ti=\frac {\psi'(\tau/2)}{2\psi(1/3+\tau/2)}.\end{equation}
We now let $z\rightarrow\tau/2$ in the identity
$$\frac{\psi(z)}{1-p^2e^{4\pi \ti z}}=\frac{p^{\frac 1{12}}e^{-\pi \ti z}(p^2,e^{4\pi \ti z},p^4e^{-4\pi \ti z};p^2)_\infty}{\tha(-e^{2\pi \ti z};p^2)_\infty}, $$
to obtain
$$\psi'\left(\frac\tau 2\right)=4\pi \ti\frac{p^{-\frac 5{12}}(p^2;p^2)_\infty^3}{\tha(-p;p^2)}. $$
Combining this with \eqref{psd}  gives  \eqref{ld2}.
\end{proof}

\begin{lemma}\label{xdl}In the notation above,
\begin{equation}\label{fd}2\frac{\psi'}{\psi}\, x'-12\pi \ti\dot  x=-\frac{\chi}{\zeta+2}\, B, \end{equation}
where
$$B(x,\zeta)=( x-1)\big((\zeta+2) x+2\zeta+1\big).$$
\end{lemma}

\begin{proof}
Let $q$ denote the left-hand side of \eqref{fd}. 
It is easy to check
that $q$ is an even elliptic function with periods $1$, $\tau$.
Thus, as a function of $z$, it 
 is a rational function of $ x(z)$.  By \eqref{xp},
 $ x'$ vanishes at all zeroes of $\psi$, so $q$ can have poles only where $ x$ has poles. This means that $q$ is  a polynomial in $ x$. 
Moreover, since $ x$ has only single poles, $q$ has at most double poles, which means that $q$ is a polynomial  of  degree at most $2$, say $q( x)=\alpha x^2+\beta x+\gamma$. 

 Since $ x(1/3)=0$, it is clear that
 $\dot x(1/3)=0$, and thus
$$\gamma=q(0)=q( x(1/3))=2\frac{\psi'}{\psi}(1/3) x'(1/3). $$
Using \eqref{ld} and \eqref{xp}, one readily computes
$$\gamma=4\pi^2(p^2;p^2)_\infty^4\tha(-p;p^2)\theta(-p\om;p^2)^3.$$ 

Next, differentiating the identity $ x^{-1}(1/3+\tau/2)=0$ gives
$$\frac{x'+2\dot x}{x^2}\Bigg|_{z=\frac 13+\frac\tau2}=0.$$ 
It follows that
$$\alpha=\lim_{z\rightarrow \frac 13+\frac\tau2}\frac1{x^2}\left(2\frac{\psi'}{\psi}x'-12\pi \ti \dot x\right)
=\lim_{z\rightarrow \frac 13+\frac\tau2}2\frac{x'}{x^2}\left(\frac{\psi'}\psi+3\pi \ti\right). $$
Using \eqref{ld2} and \eqref{xp}, this can be simplified to  
$\alpha=-\chi.$
By  \eqref{zp}, it follows that
$$\frac \gamma\alpha=-p^{-1}\frac{\tha(-p;p^2)\tha(-p\om;p^2)^3}{\tha(-1;\om)\tha(-\om;p^2)^3}=-\frac{2\zeta+1}{\zeta+2}. $$

Finally, we let $z=1/2$. Since $ x(1/2)=1$, $\dot x(1/2)=0$. 
Moreover, it is clear from \eqref{xp} that $ x'(1/2)=0$, while $\psi(1/2)\neq 0$. This shows that $q( x(1/2))=q(1)=0$.  We conclude that indeed
$$q=\alpha( x-1)\left( x-\frac \gamma\alpha\right)=-\chi( x-1)\left( x+\frac{2\zeta+1}{\zeta+2}\right). $$
\end{proof}


\begin{lemma}\label{zdl} The function $\dot\zeta=\partial\zeta/\partial\tau$ can be expressed as
$$12\pi \ti\dot\zeta=\chi(\zeta-1)(2\zeta+1). $$
\end{lemma}

\begin{proof}
Let $z=0$ in \eqref{fd}. By  \eqref{xp}, \eqref{zp} and \eqref{pdn},
$$\frac{\psi' x'}{\psi}(0)=\frac{8\pi^2\om(p^2;p^2)_\infty^4\tha(p;p^2)\tha(-p\om;p^2)^2}{\tha(-\om;p^2)\tha(p\om;p^2)^3}=-\frac{2\chi(\zeta+1)(2\zeta+1)}{\zeta+2}.$$
Since $ x(0)=2\zeta+1$, we have $\dot x(0)=2\dot\zeta$ and
$$B(x(0),\zeta)=2\zeta(2\zeta+1)(\zeta+3). $$
Combining these facts, we find that 
$$12\pi \ti\dot\zeta=\chi\left(\frac{\zeta(2\zeta+1)(\zeta+3)}{\zeta+2}-\frac{2(\zeta+1)(2\zeta+1)}{\zeta+2}\right),$$
which simplifies to the desired result.
\end{proof}

\begin{lemma}\label{pdl}
One may write
$$(\log\psi)''=C(\tau)+\frac{\chi}{\zeta+2}\,D, $$
where $C$ is independent of $z$ and
$$D(x,\zeta)=\frac{(x-\zeta)\big(x(\zeta+2)+\zeta(2\zeta+1)\big)^2}{(x-(2\zeta+1))\big((\zeta+2)x-\zeta\big)\big((\zeta+2)x-\zeta(2\zeta+1)\big)}.$$
\end{lemma}

\begin{proof}
It is easy to see that both  $(\log\psi)''$ 
and $D$ are even elliptic function with periods $1$ and $\tau$, the only singularities being double poles at the zeroes of $\psi$.
If we can show that $\lim_{\psi(z)\rightarrow 0}\psi^2((\log\psi)''-\chi D/(\zeta+2))=0$, then the conclusion follows from Liouville's theorem.   

Clearly,
$$\lim_{\psi(z)\rightarrow 0}\psi^2(\log\psi)''=-\lim_{\psi(z)\rightarrow 0}(\psi')^2. $$
If we let $P$ and $Q$ denote the numerator and denominator of $D$, respectively, then l'H\^{o}pital's rule and \eqref{sd} give
\begin{equation*}\begin{split}\lim_{\psi(z)\rightarrow 0}\frac{\psi^2\chi D}{\zeta+2}&=\lim_{\psi(z)\rightarrow 0}\frac{2 (\psi')^2\chi P}{(\zeta+2)Q''}=-\lim_{\psi(z)\rightarrow 0}\frac{8\zeta(\zeta+1)(\psi')^2P}{\left(2a\frac{\partial^2 Q}{\partial x^2}+\frac{\partial a}{\partial x}\frac{\partial Q}{\partial x}\right)}\\
&=-\lim_{\psi(z)\rightarrow 0}\frac{8\zeta(\zeta+1)(\psi')^2P}{\frac{\partial a}{\partial x}\frac{\partial Q}{\partial x}}.
\end {split}\end{equation*}
We are now reduced to  verifying the polynomial identity
$$8\zeta(\zeta+1)P=\frac{\partial a}{\partial x}\frac{\partial Q}{\partial x} $$
at the three points $x=2\zeta+1$, $x=\zeta/(\zeta+2)$ and $x=\zeta(2\zeta+1)/(\zeta+2)$, corresponding to the three zeroes modulo $\mathbb Z+\tau\mathbb Z$
of $\psi$.
\end{proof}

\begin{lemma}\label{pfl}
The modified potential \eqref{pot} can be expressed as
$$V(z)=\frac{\chi}{2\zeta(\zeta+1)(\zeta+2)}\,W(x), $$
where $W$ is as in \eqref{wp}.
\end{lemma}

\begin{proof}
Although it is straightforward to check this from \eqref{xd} and  \eqref{xv}, we will 
use a different method. By Liouville's theorem,  the  first term in \eqref{pot} can be written
$$\frac{k_0(k_0+1)}{\phi(z)^2}
=C\frac{(x-1)(x-\zeta)^2}{x^2(x-(2\zeta+1))},
 $$
with $C$ independent of $z$. 
We rewrite this
as
$$C{\phi(z)^2}=k_0(k_0+1)\frac{x^2(x-(2\zeta+1))}{(x-1)(x-\zeta)^2}, $$
and apply  $\partial^2/\partial z^2$ at the point 
 $z=1/3$ to both sides. 
Using \eqref{xds} and the fact that $\phi(1/3)=0$ and $\phi'(1/3)^2=1$, we obtain
$$2C=2k_0(k_0+1)x'(1/3)^2
\frac{(x-(2\zeta+1))}{(x-1)(x-\zeta)^2}\Bigg|_{x=0}=-k_0(k_0+1)\frac{\chi(2\zeta+1)^3}{\zeta(\zeta+1)(\zeta+2)}.$$
This gives the first term in the expression for $W$. The other terms 
can be treated similarly, or be derived from the first term using
 \cite[Lemma 2.7]{r1a}.
\end{proof}

We can now  write \eqref{sem} in algebraic form.
We choose the function $\Psi$ in Theorem \ref{semt} as in \eqref{eps}.
We express the left-hand side of \eqref{sem} using \eqref{pu},  and then
apply Lemmas \ref{ddl},  \ref{xdl},  \ref{zdl},  \ref{pdl} and  \ref{pfl}. The term involving the constant 
$C(\tau)$ from Lemma \ref{pdl} is moved to the right-hand side. 
Finally, we multiply the resulting equation through with $-2\zeta(\zeta+1)(\zeta+2)/\chi$. We find that, up to a factor independent of the variables $x_j$,
\begin{multline}\label{dnc}\left(\sum_{j=1}^m \left(a(x_j)\frac{\partial^2}{\partial x_j^2}+b(x_j)\frac{\partial}{\partial x_j}+C(x_j)\right)+md\frac{\partial}{\partial\zeta}\right)\prod_{j=1}^mF(x_j)\Delta T_n^{(\mathbf k)}\\
\sim\prod_{j=1}^mF(x_j)\Delta T_n^{(\mathbf k)}, \end{multline}
where $a$ and $d$ are as in Theorem \ref{pdet},
$$
b=\frac 12 \frac{\partial a}{\partial x}+2m\zeta(\zeta+1)B+2a\frac{\partial E/\partial x}{E},$$
which agrees with \eqref{b} and $C=W+C_0$, with
\begin{align*}
C_0&=2m(m-1)\zeta(\zeta+1)D+a\frac{\partial^2 E/\partial x^2}{E}\\
&\quad+\left(\frac 12 \frac{\partial a}{\partial x}+2m\zeta(\zeta+1)B\right)\frac{\partial E/\partial x}{E}+m d\frac{\partial E/\partial \zeta}{E}.\end{align*}
One may check that, with $c_0$ as in \eqref{cn}, $C_0-c_0$ is independent of $x$. 
Thus,   \eqref{dnc}
can be written
\begin{equation}\label{usm}
\Omega\prod_{j=1}^mF(x_j)\Delta(x_1,\dots,x_m) T_n^{(\mathbf k)}(x_1,\dots,x_m)=0,
\end{equation}
where 
\begin{equation}\label{rom}\Omega=\sum_{j=1}^m \left(a(x_j)\frac{\partial^2}{\partial x_j^2}+b(x_j)\frac{\partial}{\partial x_j}+c(x_j)\right)+m\, d\frac{\partial}{\partial\zeta}+e, \end{equation}
for some yet unknown function $e$ of $\zeta$.
To prove Theorem \ref{pdet}, it remains to show that
$e$ is given by \eqref{ee}.

\subsection{The constant term}
\label{cts}

To compute the constant term $e$ in \eqref{rom}, 
we will first prove \eqref{ee} for $\mathbf k=(0,0,0,0)$, 
and 
then proceed by induction on $\sum_j |k_j|$.
For both the base case and the induction step, 
our method is based on investigating limits
of \eqref{usm} when all the variables coincide.

For the proof of the next lemma, we
 will need the elementary identity 
\begin{equation}\label{dd}\sum_{j=1}^m x_j^k \frac{\partial^k}{\partial x_j^k}\,\Delta(\mathbf x)=k!\binom{m}{k+1}\Delta(\mathbf x).\end{equation}
To see this, note that
 the left-hand side is an anti-symmetric homogeneous polynomial of the same degree as $\Delta(\mathbf x)$, and thus proportional to $\Delta(\mathbf x)$. The value of the constant follows since the coefficient of $x_2x_3^2\dotsm x_m^{m-1}$ on the left-hand side is
$$\sum_{j=1}^m(j-1)(j-2)\dotsm(j-k)=k!\binom{m}{k+1}. $$

\begin{lemma}\label{sdl}
Suppose that $P$ is a symmetric formal power series in $m$ variables, whose Taylor expansion at $0$ is given by
$$\alpha+\beta\sum_{j=1}^{m}x_j+\gamma\sum_{j=1}^{m}x_j^2+\delta\sum_{1\leq j<k\leq m}x_jx_k+\text{\emph{higher order terms}} $$
and let $f$ be a formal power series
$$f(x)=a+bx+cx^2+\text{\emph{higher order terms}}. $$
Then, 
\begin{align}\nonumber\operatorname{C.T.}\frac 1{\Delta(\mathbf x)}\sum_{j=1}^m f(x_j)\frac{\partial^2}{\partial x_j^2}\,\Delta(x)P(x)&=2\binom m3 c\alpha+2\binom m2b\beta \\
\label{sda}&\quad+2m^2a\gamma -2\binom m2a\delta ,\\
\label{sdb}\operatorname{C.T.}\frac 1{\Delta(\mathbf x)}\sum_{j=1}^m f(x_j)\frac{\partial}{\partial x_j}\,\Delta(x)P(x)&=\binom m2b\alpha+ma\beta,
\end{align}
where $\operatorname{C.T.}$ stands for the constant term.
\end{lemma}

\begin{proof}
 We split the left-hand side of \eqref{sda} as
$$\operatorname{C.T.}\frac 1{\Delta}\sum_{j=1}^m f(x_j)\left(\frac{\partial^2\Delta}{\partial x_j^2}\,P(x)+2
\frac{\partial\Delta}{\partial x_j}\frac{\partial P}{\partial x_j}+\Delta\frac{\partial^2 P}{\partial x_j^2}\right)=S_1+S_2+S_3.$$
By homogeneity, only the quadratic term in $f$ contributes to $S_1$. 
It then follows from  \eqref{dd} that $S_1=2\binom m3c\alpha  $. In $S_2$, we get  contributions from the linear and constant terms in $f$.
Again by \eqref{dd}, the linear term contributes  $2\binom m2b\beta $, whereas the constant term  contributes
\begin{multline*}\operatorname{C.T.}\frac {2a}{\Delta}\sum_{j=1}^m 
\frac{\partial\Delta}{\partial x_j}\frac{\partial}{\partial x_j}\left(\gamma\sum_{k=1}^{m}x_k^2+\delta\sum_{1\leq k<l\leq m}x_kx_l\right)\\
=
\operatorname{C.T.}\frac {2a}{\Delta}\sum_{j=1}^m 
\frac{\partial\Delta}{\partial x_j}\left((2\gamma-\delta) x_j+\delta\sum_{k=1}^m x_k\right)=2\binom m2a(2\gamma-\delta),
\end{multline*}
where we used  that
 $\sum_{j=1}^m \partial\Delta/\partial x_j=0$ (as it is an anti-symmetric homogeneous polynomial of lower degree than $\Delta$). Finally, $S_3=2m a\gamma$, which completes the proof of \eqref{sda}.
The proof of \eqref{sdb} is similar.
\end{proof}

We want to apply Lemma \ref{sdl} to the case 
$\mathbf k=(0,0,0,0)$ of
\eqref{usm}. For this purpose, we  need  
the lowest terms in the Taylor expansion of $T_n^{(0,0,0,0)}$.

\begin{lemma}\label{btl}
For $m=2n\geq 1$, the Taylor expansion of $T_n^{(0,0,0,0)}(x_1,\dots,x_{m})$ around $x_1=\dots= x_{m}=0$ has the form
$$\alpha_n+\beta_n\sum_{j=1}^{m}x_j+\gamma_n\sum_{j=1}^{m}x_j^2+\delta_n\sum_{1\leq j<k\leq m}x_jx_k+\text{\emph{higher order terms}}, $$
where
\begin{align*}
\alpha_n&=\zeta^{n(n-1)}(2\zeta+1)^{n(n-1)}\\
\beta_n&=-(n-1)\zeta^{n(n-1)}(2\zeta+1)^{n^2-n-1},\\
\gamma_n&=\frac{(n-1)(n-2)}{2}\,\zeta^{n(n-1)}(2\zeta+1)^{(n+1)(n-2)},\\
\delta_n&=(n-1)^2\zeta^{n(n-1)}(2\zeta+1)^{(n+1)(n-2)}.
\end{align*} 
\end{lemma}

\begin{proof}
Expand \eqref{gdf} along the first row and then let
 $x_1=x_{n+1}=0$. Then, only the first term gives a non-zero contribution. Rewriting the complementary minor in terms of $T_{n-1}^{(0,0,0,0)}$ gives after relabelling the parameters
$$T_n^{(0,0,0,0)}(x_1,\dots,x_{2n-2},0,0)=\zeta^{2n-2}\prod_{j=1}^{2n-2}(2\zeta+1-x_j) 
T_{n-1}^{(0,0,0,0)}(x_1,\dots,x_{2n-2}).$$
This leads to the  system of recursions
\begin{align*} \alpha_n&=\zeta^{2n-2}(2\zeta+1)^{2n-2}\alpha_{n-1},\\
\beta_n&=\zeta^{2n-2}(2\zeta+1)^{2n-3}\big((2\zeta+1)\beta_{n-1}-\alpha_{n-1}\big),\\
\gamma_n&=\zeta^{2n-2}(2\zeta+1)^{2n-3}\big((2\zeta+1)\gamma_{n-1}-\beta_{n-1}\big),\\
\delta_n&=\zeta^{2n-2}(2\zeta+1)^{2n-4}\big((2\zeta+1)^2\delta_{n-1}-2(2\zeta+1)\beta_{n-1}+\alpha_{n-1}\big),\end{align*}
which is easily solved from the initial value $\alpha_1=1$, $\beta_1=\gamma_1=\delta_1=0$.
\end{proof}

We can now prove \eqref{usm} for $\mathbf k=(0,0,0,0)$. This result
has  been obtained by Zinn-Justin \cite[\S 4.2.2]{zj}, see \cite[\S 5.3]{r1a} for the precise relation to the notation used there.

\begin{lemma}[Zinn-Justin]\label{zjl}
\emph{Theorem \ref{pdet}} holds for $\mathbf k=(0,0,0,0)$.
\end{lemma}

\begin{proof}
Applying Lemma \ref{sdl} to the case  $\mathbf k=(0,0,0,0)$ of \eqref{usm} gives
\begin{multline*}\left(\binom m 3 \frac{\partial^2 a}{\partial x^2}+\binom m2 \frac{\partial b}{\partial x}+mc_0+e\right)\Bigg|_{x=0}\alpha_{n}
+\left(2\binom m2 \frac{\partial a}{\partial x}+mb\right)\Bigg|_{x=0}\beta_n\\
+a(0)\left(2m^2 \gamma_n-2\binom m2\delta_n\right)+md \frac{\partial \alpha_n}{\partial \zeta}=0,
 \end{multline*}
where $m=2n$. Inserting the expressions given in Theorem \ref{pdet} and
Lemma \ref{btl} yields $e=0$, in agreement with \eqref{ee}.
\end{proof}

It seems difficult to extend this proof to general
 $\mathbf k$. To proceed, we write \eqref{usm} 
as in \eqref{usf} (with $e$ still unknown), and 
 then  let all variables $x_j$ tend to $\xi_l$. Although, in general, $b_F$, $c_F$ and $T_n^{(\mathbf k)}$ have poles, 
they are regular   at the point
$\xi_l$. Thus, we may apply Lemma \ref{sdl}, with $x_j$ replaced by $x_j-\xi_l$. Since $a(\xi_l)=0$, the result simplifies to
\begin{subequations}\label{usls}
\begin{multline}\label{usl}\left(\binom m 3 \frac{\partial^2 a}{\partial x^2}+\binom m2 \frac{\partial b_F}{\partial x}+mc_F+e\right)\Bigg|_{x=\xi_l}\alpha
+\left(2\binom m2 \frac{\partial a}{\partial x}+mb_F\right)\Bigg|_{x=\xi_l}\beta\\
+md\varepsilon=0,
 \end{multline}
 where
$$\alpha=T_n^{(\mathbf k)}(\xi_l^{(m)}),\qquad \beta=\frac{\partial T_n^{(\mathbf k)}}{\partial x_1}(\xi_l^{(m)}),\qquad \varepsilon=\frac{\partial T_n^{(\mathbf k)}}{\partial \zeta}(\xi_l^{(m)}). $$

Consider now 
 \eqref{usl}, with $\mathbf k$  replaced by $\mathbf k-\mathbf e_l$ (with $\mathbf e_l$ 
a unit vector) and $m$ by $m+1$. 
By \eqref{ten},
$$T_n^{(\mathbf k)}(x_1,\dots,x_m)=T_n^{(\mathbf k-\mathbf e_l)}(x_1,\dots,x_m,\xi_l). $$
It follows that, as the indices change,
 $\alpha$ and $\beta$ remain the same whereas $\varepsilon$ is replaced by $\varepsilon-\beta\partial\xi_l/\partial\zeta$. 
Thus,
\begin{multline}\label{uslb}\left(\binom {m+1} 3 \frac{\partial^2 a}{\partial x^2}+\binom {m+1}2 \frac{\partial \tilde b_F}{\partial x}+(m+1)\tilde c_F+\tilde e\right)\Bigg|_{x=\xi_l}\alpha\\
+\left(2\binom {m+1}2 \frac{\partial a}{\partial x}+(m+1)\tilde b_F-(m+1)d\frac{\partial\xi_l}{\partial\zeta}\right)\Bigg|_{x=\xi_l}\beta
+(m+1)d\varepsilon=0,
 \end{multline}
\end{subequations}
where $\sim$ signifies the change $(\mathbf k,m)\mapsto (\mathbf k-\mathbf e_l,m+1)$
in the coefficients depending on these indices.

Eliminating $\varepsilon$ from the equations \eqref{usls},
the resulting coefficient of $\beta$ is
\begin{multline*}
(m+1)\left(2\binom m2 \frac{\partial a}{\partial x}+mb_F\right)\\
-m\left(2\binom {m+1}2 \frac{\partial a}{\partial x}+(m+1)\tilde b_F-(m+1)d\frac{\partial\xi_l}{\partial\zeta}\right)
\Bigg|_{x=\xi_l}=0,
\end{multline*}
by a direct computation.
  Since we know from \cite[Cor.\ 3.9]{r1a} that
$\alpha=T_n^{(\mathbf k+me_l)}$ does not vanish identically, it follows that
\begin{multline}\label{eer}
(m+1)\left(\binom m 3 \frac{\partial^2 a}{\partial x^2}+\binom m2 \frac{\partial b_F}{\partial x}+mc_F+e\right)\\
-m\left(\binom {m+1} 3 \frac{\partial^2 a}{\partial x^2}+\binom {m+1}2 \frac{\partial \tilde b_F}{\partial x}+(m+1)\tilde c_F+\tilde e\right)
\Bigg|_{x=\xi_l}=0.
\end{multline}
We view this as a recursion for obtaining the unknown coefficient $e$ from $\tilde e$. By another direct computation, it is consistent with the explicit expression \eqref{ee}. This proves the following induction step.

\begin{lemma}\label{url}
If \emph{Theorem \ref{pdet}} holds for fixed   $\mathbf k$ and $m\geq 1$, then it also holds when $\mathbf k$ is replaced by $\mathbf k+\mathbf e_l$ and $m$ by $m-1$.
\end{lemma}

We will need another recursion, which follows from
the following differential equation for the polynomials $U_n^{(\mathbf k)}$ defined in \eqref{unk}.

\begin{proposition}\label{sup}
The polynomials  $U_n^{(\mathbf k)}$  satisfy the differential equation
\begin{equation}\label{usu}
\Om\prod_{j=1}^mK(x_j)\Delta(x_1,\dots,x_m) U_n^{(\mathbf k)}(x_1,\dots,x_m)=0,
\end{equation}
where $\Om$ is as in \eqref{rom}, and $K=1/F\sqrt a$. 
\end{proposition}

The main point of Proposition \ref{sup} is that 
\eqref{usu} holds with  the same $e$ 
as in \eqref{usm}. It is easier to see that
 it holds for {some}  $e$ or, equivalently, for some function $K\sim 1/F\sqrt a$ up to a $\zeta$-dependent factor.

\begin{proof}
The operator $\Omega$ has been constructed so that
$$H^{(\mathbf k)}M_n\Omega (H^{(\mathbf k)}M_n)^{-1}=\Phi^{(\mathbf k)}\mathcal H\Phi^{(-\mathbf k)}+C, $$
where 
$$H^{(\mathbf k)}=\prod_{j=1}^m\prod_{l=0}^3\big((x_j-\xi_l)G(x_j,\xi_l)\big)^{\frac {k_l}2},  $$
$M_n$ is given in \eqref{m},  
$\mathcal H$  in \eqref{om} and  $C$ is independent of the variables $z_j$.
We conjugate this identity with $\Phi^{(2\mathbf k^-)}$. Using that
$${\Phi^{(2\mathbf k^-)}H^{(\mathbf k)}M_n}\sim H ^{(\mathbf k^++\mathbf k^-)}{M_{n+|\mathbf k^-|}} $$
up to a factor independent of the variables $z_j$, we find that
(with a change of $C$)
\begin{equation}\label{pso}H^{(\mathbf k^++\mathbf k^-)}M_{n+|\mathbf k^-|}\Omega (H^{(\mathbf k^++\mathbf k^-)}M_{n+|\mathbf k^-|})^{-1}
=\Phi^{(\mathbf k^++\mathbf k^-)}\mathcal H\Phi^{(-\mathbf k^+-\mathbf k^-)}+C. 
\end{equation}

It follows from \eqref{php} that $\sigma^{\otimes m}$ commutes with the right-hand side of \eqref{pso}. If $L$ denote the left-hand side of \eqref{pso}, we 
apply \eqref{us}, with $n$
replaced by $n+|\mathbf k^-|$, to the left-hand side of 
 $L\sigma^{\otimes m }=\sigma^{\otimes m } L$. 
This gives
\begin{multline}\label{sic}H^{(\mathbf k^++\mathbf k^-)}M_{n+|\mathbf k^-|}\Omega H^{(-\mathbf k^+-\mathbf k^-)}(\sqrt a^{-1})^{\otimes m}\hat\sigma_{n+|\mathbf k^-|}^{\otimes m}\\
=\sigma^{\otimes m} H^{(\mathbf k^++\mathbf k^-)}M_{n+|\mathbf k^-|}\Omega H^{(-\mathbf k^+-\mathbf k^-)}, \end{multline}
which holds on the domain of $\hat\sigma_{n+|\mathbf k^-|}^{\otimes m}$.
If we  act with \eqref{sic} on
$$H^{(\mathbf k^++\mathbf k^-)}F^{\otimes m}\Delta T_n^{(\mathbf k)},$$
 the right-hand side vanishes by \eqref{usm}. 
We can then deduce \eqref{usu} from
\eqref{stu}.
\end{proof}

We can now repeat the analysis leading to  \eqref{eer}, using \eqref{uen} and \eqref{usu} rather than \eqref{ten} and \eqref{usm}. We 
find that 
 \begin{multline*}
(m+1)\left(\binom m 3 \frac{\partial^2 a}{\partial x^2}+\binom m2 \frac{\partial b_K}{\partial x}+mc_K+e\right)\\
-m\left(\binom {m+1} 3 \frac{\partial^2 a}{\partial x^2}+\binom {m+1}2 \frac{\partial \tilde b_K}{\partial x}+(m+1)\tilde c_K+\tilde e\right)
\Bigg|_{x=\xi_l}=0.
\end{multline*}
where $\sim$ now denotes  the change of indices 
$(\mathbf k,m)\mapsto (\mathbf k+\mathbf e_l,m+1)$.
Again, this is consistent with \eqref{ee}, which proves the following lemma.

\begin{lemma}\label{drl}
If \emph{Theorem \ref{pdet}} holds for fixed   $\mathbf k$ and $m\geq 1$, it also holds when $\mathbf k$ is replaced by $\mathbf k-\mathbf e_l$ and $m$ by $m-1$.
\end{lemma}

It is clear that Lemmas \ref{zjl}, \ref{url} and \ref{drl} together imply 
Theorem \ref{pdet} by induction on $\sum_j|k_j|$.

\section{Bilinear identities}
\label{bis}

The case $m=2n-|\mathbf k|=0$, when $T_n^{(\mathbf k)}$ depends only on $\zeta$, is of particular interest. 
We will write
$t^{(\mathbf k)}=T_{|\mathbf k|/2}^{(\mathbf k)}$.  
 In the subsequent paper \cite{r2}, we will show that $t^{(\mathbf k)}$ can be identified with tau functions of Painlev\'e VI. 
We will now explain how bilinear identities for tau functions arise from our construction. Rather than giving a complete list, we will 
just give two examples of such relations, which will in fact be used in \cite{r2} to
obtain the identification with tau functions. 
Once this idenfication has been established, one can obtain further
bilinear identities from Painlev\'e theory; some examples
 are discussed in \cite{r2}.

The point of the following result is to 
characterize $t^{(\mathbf k)}$ by a very short list of properties.
 Recall from
\cite[Cor.\ 2.19 and Cor.\ 2.21]{r1a} that the lattice of functions 
$t^{(\mathbf k)}$ is symmetric under the group $G=\mathrm S_4\times \mathrm S_2$ in the following sense. If  $\mathrm S_4$ acts by permuting $(k_0,k_1,k_2,k_3)$  and $\mathrm S_2$ by the reflection
 $(k_0,k_1,k_2,k_3)\mapsto(-k_0-1,-k_1-1,-k_2-1,-k_3-1)$, then for any $\sigma\in  G$ there holds an identity
\begin{equation}\label{tsy}t^{(\mathbf k)}(\zeta)=\phi(\zeta)t^{(\sigma \mathbf k)}(\psi(\zeta)),\end{equation}
with $\phi=\phi_{\mathbf k,\sigma}$ and $\psi=\psi_\sigma$  rational functions that can be given explicitly.

\begin{theorem}\label{rt}
The functions $t^{(\mathbf k)}$ satisfy the two identities
\begin{subequations}\label{km}
\begin{multline}\label{kma}
t^{(\mathbf k-2\mathbf e_0)}t^{(\mathbf k+\mathbf e_0+\mathbf e_1)}=\zeta^2(\zeta+1)(\zeta-1)(2\zeta+1)^2\\
\times\left(\frac{1}{2k_0-1}\,t^{(\mathbf k)}\frac{d t^{(\mathbf k-\mathbf e_0+\mathbf e_1)}}{d \zeta}-\frac{1}{2k_0+1}\frac{d t^{(\mathbf k)}}{d \zeta} t^{(\mathbf k-\mathbf e_0+\mathbf e_1)}\right)\\
+\frac{\zeta(2\zeta+1)}{2(2k_0-1)(2k_0+1)(\zeta+2)}A^{(\mathbf k)}t^{(\mathbf k)}t^{(\mathbf k-\mathbf e_0+\mathbf e_1)},
\end{multline}
\begin{multline}\label{kmb}
t^{(\mathbf k-2\mathbf e_0)}t^{(\mathbf k+\mathbf e_0-\mathbf e_1)}=\frac{(\zeta+1)(\zeta-1)(2\zeta+1)^2(\zeta+2)^2}{\zeta^2}\\
\times\left(\frac{1}{2k_0-1}\,t^{(\mathbf k)}\frac{d t^{(\mathbf k-\mathbf e_0-\mathbf e_1)}}{d \zeta}-\frac{1}{2k_0+1}\frac{d t^{(\mathbf k)}}{d \zeta} t^{(\mathbf k-\mathbf e_0-\mathbf e_1)}\right)\\
+\frac{(2\zeta+1)(\zeta+2)}{2(2k_0-1)(2k_0+1)\zeta^3}B^{(\mathbf k)}t^{(\mathbf k)}t^{(\mathbf k-\mathbf e_0-\mathbf e_1)},
\end{multline}
\end{subequations}
where
{\allowdisplaybreaks
\begin{align*}
A^{(\mathbf k)}&=(2\zeta^4-23\zeta^3-36\zeta^2-5\zeta+8)k_0^2\\
&\quad-\zeta(2\zeta+1)(3\zeta^2+10\zeta+5)k_1(2k_0+k_1)\\
&\quad-\zeta(6\zeta^3+19\zeta^2+4\zeta-11)k_2^2\\
&\quad-\zeta(2\zeta+1)(3\zeta^2+2\zeta+1)k_3^2\\
&\quad-2\zeta(\zeta-1)(2\zeta+1)(\zeta+3)(k_0+k_1)k_2\\
&\quad-2(\zeta-1)(2\zeta+1)(3\zeta^2+9\zeta+4)(k_0+k_1)k_3\\
&\quad-2(2\zeta+1)(\zeta^3+6\zeta^2+3\zeta-4)k_2k_3\\
&\quad-4(2\zeta+1)(\zeta^2+5\zeta+3)(k_0+k_1)\\
&\quad+4(2\zeta+1)(2\zeta^3+5\zeta^2-\zeta-3)k_2\\
&\quad+4(2\zeta+1)(\zeta^2+\zeta+1)k_3\\
&\quad-4(\zeta+1)^2(2\zeta^2-\zeta+2),\\
B^{(\mathbf k)}&=(10\zeta^4+13\zeta^3-28\zeta^2-41\zeta-8)k_0^2\\
&\quad-\zeta(2\zeta+1)(3\zeta^2+10\zeta+5)k_1(k_1-2k_0)\\
&\quad-\zeta(6\zeta^3+19\zeta^2+4\zeta-11)k_2^2\\
&\quad-\zeta(2\zeta+1)(3\zeta^2+2\zeta+1)k_3^2\\
&\quad+2\zeta(\zeta-1)(2\zeta+1)(\zeta+3)(k_0-k_1)k_2\\
&\quad+2(\zeta-1)(2\zeta+1)(3\zeta^2+9\zeta+4)(k_0-k_1)k_3\\
&\quad-2(2\zeta+1)(\zeta^3+6\zeta^2+3\zeta-4)k_2k_3\\
&\quad+2(\zeta-1)(2\zeta+1)(\zeta+3)(3\zeta+2)(k_1-k_0)\\
&\quad+2(2\zeta+1)(5\zeta^3+12\zeta^2-5\zeta-6)k_2\\
&\quad+2(2\zeta+1)(3\zeta^3+8\zeta^2-3\zeta-2)k_3\\
&\quad-2(8\zeta^4+18\zeta^3-7\zeta^2-18\zeta-4).
\end{align*}
}
Moreover, the lattice of functions $t^{(\mathbf k)}$, where $\mathbf k=(k_0,k_1,k_2,k_3)\in\mathbb Z^4$ with $\sum_j k_j$ even, is uniquely determined by \eqref{tsy}, \eqref{km}
and the three  values 
\begin{align*}
t^{(0,0,0,0)}&=t^{(1,-1,0,0)}=1,\\
t^{(0,-1,-1,0)}&=-\frac{2\zeta^2(\zeta-1)(\zeta+1)^2(2\zeta+1)}{(\zeta+2)^2}.
\end{align*}
\end{theorem}

\begin{proof}
We start from  the Jacobi--Desnanot identity in the form
\cite[Eq.\ (2.41a)]{r1a}
\begin{multline*}(a-b)(c-d)T(\mathbf x;\mathbf y)T(a,b,c,d,\mathbf x;\mathbf y)
=G(a,d)G(b,c)T(a,c,\mathbf x;\mathbf y)T(b,d,\mathbf x;\mathbf y)\\
-G(a,c)G(b,d)T(a,d,\mathbf x;\mathbf y)T(b,c,\mathbf x;\mathbf y).
\end{multline*}
When
$b=c=\xi_0$, $d=\xi_1$, $\mathbf x=\boldsymbol\xi^{(\mathbf k^+)}$ and $y=\boldsymbol\xi^{(\mathbf k^-)}$,
this can be written
\begin{multline*}(a-\xi_0)(\xi_0-\xi_1)t^{(\mathbf k)}T_{n+2}^{(\mathbf k+2\mathbf e_0+\mathbf e_1)}(a)
\\
=G(a,\xi_1)G(\xi_0,\xi_0)T_{n+1}^{(\mathbf k+\mathbf e_0)}(a)t^{(\mathbf k+\mathbf e_0+\mathbf e_1)}-
G(a,\xi_0)G(\xi_0,\xi_1)
t^{(\mathbf k+2\mathbf e_0)}T_{n+1}^{(\mathbf k+\mathbf e_1)}(a),
\end{multline*}
where $|\mathbf k|=2n$.
Differentiating  with respect to $a$ and  letting  $a=\xi_0$ gives
\begin{multline}\label{tpr}(\xi_0-\xi_1)t^{(\mathbf k)}t^{(\mathbf k+3\mathbf e_0+\mathbf e_1)}
\\
=\left(G(\xi_0,\xi_0)\frac{\partial G}{\partial x}(\xi_0,\xi_1)-
\frac{\partial G}{\partial x}(\xi_0,\xi_0)G(\xi_0,\xi_1)\right)t^{(\mathbf k+2\mathbf e_0)}t^{(\mathbf k+\mathbf e_0+\mathbf e_1)}\\
+G(\xi_0,\xi_0)G(\xi_0,\xi_1)\left(\frac{\partial T_{n+1}^{(\mathbf k+\mathbf e_0)}}{\partial x}(\xi_0)t^{(\mathbf k+\mathbf e_0+\mathbf e_1)}-t^{(\mathbf k+2\mathbf e_0)}\frac{\partial T_{n+1}^{(\mathbf k+\mathbf e_1)}}{\partial x}(\xi_0)\right).
\end{multline}

The main point is now that the specialized derivatives of $T$-functions in \eqref{tpr} can be expressed  in terms of $t$-functions using the Schr\"odinger equation.
The relevant identity is a special case of \eqref{uslb}, but for clarity we repeat the argument. If we let $x_1\rightarrow\xi_0$
in the case $m=1$ of \eqref{usf}, we get
$$b_F(\xi_0)\frac{\partial T_n^{(\mathbf k)}}{\partial x}(\xi_0)+ (c_F(\xi_0)+e)T_n^{(\mathbf k)}(\xi_0)+d\frac{\partial T_n^{(\mathbf k)}}{\partial \zeta}(\xi_0)=0. $$
 On the other hand, differentiating the equality $T_n^{(\mathbf k)}(\xi_0)=t^{(\mathbf k+\mathbf e_0)}$ gives
$$2\frac{\partial T_n^{(\mathbf k)}}{\partial x}(\xi_0)+\frac{\partial T_n^{(\mathbf k)}}{\partial \zeta}(\xi_0)=\frac{d t^{(\mathbf k+\mathbf e_0)}}{d \zeta}. $$
Eliminating $\partial T_n^{(\mathbf k)}/\partial\zeta$ from these two equations, we find that
\begin{equation}\label{bst}\frac{\partial T_n^{(\mathbf k)}}{\partial x}(\xi_0)=\frac 1{2d-b_F(\xi_0)}\left((c_F(\xi_0)+e)t^{(\mathbf k+\mathbf e_0)}+d\frac{d t^{(\mathbf k+e_0)}}{d \zeta}\right). \end{equation}
Using \eqref{bst} on the right-hand side of \eqref{tpr}
 gives, after replacing $k_0$ by $k_0-2$, \eqref{kma}.
The identity \eqref{kmb} is proved similarly, starting instead from
 \cite[Eq.\ (2.41b)]{r1a}
\begin{multline*}(a-b)T(\mathbf x;\mathbf y)T(a,b,c,\mathbf x;d,\mathbf y)
=(a-d)G(a,d)G(b,c)T(a,c,\mathbf x;\mathbf y)T(b,\mathbf x;d,\mathbf y)\\
-(b-d)G(a,c)G(b,d)T(a,\mathbf x;d,\mathbf y)T(b,c,\mathbf x;\mathbf y).
\end{multline*}

To show that $t^{(\mathbf k)}$ can be constructed from the given data,
we apply induction on 
$N(\mathbf k)=\sum_{j=0}^3|k_j+1/2|$. Thus, fixing $\mathbf k$, suppose that $t^{(\mathbf l)}$ is known for all $\mathbf l$ with $N(\mathbf l)<N(\mathbf k)$. 
By the symmetries \eqref{tsy} and the fact that
 $N(\mathbf k)$ is invariant under the group action, 
we  may replace $\mathbf k$ by any element in the same orbit.
We choose this element so that
$k_0+1/2\geq\max_{1\leq j\leq 3}|k_j+1/2|$. In particular, $k_0\geq 0$.
If $k_1\geq 0$, we  use \eqref{kma}, with $\mathbf k$ replaced by $\mathbf k-\mathbf e_0-\mathbf e_1$, to define $t^{(\mathbf k)}$. If $k_1\leq -1$, we  use instead
\eqref{kmb}, with $\mathbf k$ replaced by $\mathbf k-\mathbf e_0+\mathbf e_1$.
This does not lead to division by zero, since  $t^{(\mathbf k)}$ never vanishes identically
\cite[Cor.\ 3.9]{r1a}.

For the construction just described to work,  the remaining functions $t^{(\mathbf l)}$ appearing in \eqref{km} must satisfy $N(\mathbf l)<N(\mathbf k)$. Thus, in the case $k_1\geq 0$ we must have
\begin{align*}
\left|k_0-\frac 52\right|+\left|k_1-\frac 12\right|&< k_0+k_1+1,\\
\left|k_0-\frac 32\right|&< k_0+\frac 12,\\
\left|k_0-\frac 12\right|+\left|k_1-\frac 12\right|&< k_0+k_1+1,
\end{align*}
It is an elementary exercise to check that this is true except in the cases
$(k_0,k_1)=(0,0)$ and $(k_0,k_1)=(1,0)$. Similarly, when $k_1\leq 0$ the construction works except if $(k_0,k_1)=(0,-1)$ or $(k_0,k_1)=(1,-1)$. 
In conclusion, the induction step works unless 
$k_0\in\{0,1\}$, $k_1\in\{0,-1\}$. Repeating the same construction  with 
$k_1$ replaced by $k_2$ or $k_3$ we are left with the exceptional cases
 $k_0\in\{0,1\}$, $k_1,k_2,k_3\in\{0,-1\}$. 
Since $|\mathbf k|$ is even there are eight such cases, which
  split into four $G$-orbits represented by
$\mathbf k=(0,0,0,0)$,  $(0,-1,-1,0)$, $(1,-1,0,0)$, $(1,-1,-1,-1)$.
We have chosen the first three  as initial values. The case
$\mathbf k=(0,0,0,0)$ of \eqref{kmb} expresses a point in the fourth orbit
in terms of the other three.
\end{proof}

 \end{document}